\documentclass[pra,aps,twocolumn,showpacs,preprintnumbers,showkeys,superscriptaddress,epsf]{revtex4}
\usepackage{graphics,amsmath,amsfonts,amscd,revsymb,latexsym, enumerate,multirow,graphicx}

\newtheorem{theorem}{Theorem}[section]
\newtheorem{corollary}[theorem]{Corollary}
\newtheorem{lemma}[theorem]{Lemma}

\newtheorem{proposition}[theorem]{Proposition}
\newtheorem{definition}[theorem]{Definition}

\newtheorem{eg}[theorem]{Example}

\newcommand{\qed}{{\hfill$\Box$}}
\newenvironment{proof}{\noindent \textbf{{Proof~} }}{\qed}

\def\bi{\begin{itemize}}
\def\ei{\end{itemize}}
\def\be{\begin{equation}}
\def\ee{\end{equation}}
\def\bea{\begin{eqnarray}}
\def\eea{\end{eqnarray}}
\def\ben{\begin{eqnarray*}}
\def\een{\end{eqnarray*}}

\def\>{\rangle}
\def\<{\langle}

\def\bH{{\bf H}}

\newcommand{\bc}{{\mathbf c}}
\newcommand{\bd}{{\mathbf d}}

\def\bbZ{\mathbb{Z}}
\def\bbF{\mathbb{F}}
\newcommand{\1} I 

\def\*{\star}

\def\0{{\mathbf{0}}}
\def\1{{\mathbf{1}}}
\def\2{{\mathbf{2}}}
\def\3{{\mathbf{3}}}
\def\4{{\mathbf{4}}}
\def\5{{\mathbf{5}}}
\def\6{{\mathbf{6}}}
\def\7{{\mathbf{7}}}
\def\8{{\mathbf{8}}}
\def\9{{\mathbf{9}}}

\begin{document}

\title{Entanglement-Assisted Quantum Quasi-Cyclic Low-Density Parity-Check Codes}
%
\author{Min-Hsiu Hsieh, Todd A. Brun, and Igor Devetak  \\
Ming Hsieh Electrical Engineering Department \\
University of Southern California, Los Angeles, CA 90089}

\date{\today}

\begin{abstract}
We investigate the construction of quantum low-density parity-check
(LDPC) codes from classical quasi-cyclic (QC) LDPC codes with girth
greater than or equal to 6. We have shown that the classical codes in the
generalized Calderbank-Shor-Steane (CSS) construction do not need to satisfy the dual-containing property as long as pre-shared entanglement is
available to both sender and receiver. We can use this to avoid the
many 4-cycles which typically arise in dual-containing LDPC codes.
The advantage of such quantum codes comes from the use of efficient
decoding algorithms such as sum-product algorithm (SPA). It is well
known that in the SPA, cycles of length 4 make successive decoding
iterations highly correlated and hence limit the decoding
performance. We show the principle of constructing quantum QC-LDPC
codes which require only small amounts of initial shared
entanglement.
\end{abstract}


\maketitle

\section{Introduction}
Low-density parity-check (LDPC) codes were first proposed by
Gallager \cite{RG63thesis} in the early 1960s, and were rediscovered
\cite{MN96,davey98low,mackay99good} in the 90s. It has been shown
that these codes can achieve a remarkable performance that is very
close to the Shannon limit. Sometimes, they perform even better
\cite{ Mackay98turbo} than their main competitors, the Turbo codes.
These two families of codes are called modern codes.

A $(J,L)$-regular LDPC code is defined to be the null space of a
binary parity check matrix $H$ with the following properties: (1)
each column consists of $J$ ``ones'' (each column has weight $J$); (2)
each row consists of $L$ ``ones'' (each row has weight $L$); (3) both
$J$ and $L$ are small compared to the length of the code $n$ and the
number of rows in $H$. Several methods of constructing good families
of regular LDPC codes have been proposed
\cite{mackay99good,KLF01,Fossorier04}. However, probably the easiest
method is based on circulant permutation matrices
\cite{Fossorier04}, which was inspired by Gallager's original LDPC
construction.

We define a cycle of a linear code to be of length $2s$ if there is
an ordered list of $2s$ matrix elements of $H$ such that: (1) all $2s$
elements are equal to 1; (2) successive elements in the list
are obtained by alternately changing the row or column only (i.e.,
two consecutive elements will have either the same row and different
columns, or the same column and different rows); (3) the positions
of all the $2s$ matrix elements are distinct, except the first and
last ones. We call the cycle of the shortest length the girth of the
code.

There are various methods for decoding classical LDPC codes
\cite{KLF01}. Among them, \emph{sum-product algorithm} (SPA)
decoding \cite{Tanner81} provides the best trade-off between
error-correction performance and decoding complexity. It has been
shown that the performance of SPA decoding very much depends on the
cycles of shortest length \cite{Tanner81}---in particular, cycles of
length 4. These shortest cycles make successive decoding iterations
highly correlated, and severely limit the decoding performance.
Therefore, to use SPA decoding, it is important to design codes
without short cycles, especially cycles of length 4.

Because classical LDPC codes have such good performance ---
approaching the channel capacity in the limit of large block size
--- there has been considerable interest in finding quantum versions
of these codes. However, quantum low-density parity-check codes
\cite{HI07QLDPC,MMM04QLDPC,COT05QLDPC,PC08QLDPC} are far less
studied than their classical counterparts. The main obstacle comes
from the \emph{dual-containing} constraint of the classical codes
that are used to construct the corresponding quantum codes. While
this constraint was not too difficult to satisfy for relatively
small codes, it is a substantial barrier to the use of highly
efficient LDPC codes. The second obstacle comes from the bad
performance of the iterative decoding algorithm. Though the SPA
decoding can be directly used to decode the quantum errors, its
performance is severely limited by the many 4-cycles, which are
usually the by-product of the dual-containing property, in the
standard quantum LDPC codes \cite{MMM04QLDPC}.

In this paper we will show that, by using the entanglement-assisted
formalism \cite{BDH06,BDH06IEEE}, these two obstacles of standard
quantum LDPC codes can be overcome. By allowing the use of
pre-shared entanglement between senders and receivers, the
dual-containing constraint can be removed. Constructing quantum LDPC
codes from classical LDPC codes becomes transparent. That is,
arbitrary classical binary or quaternary codes can be used to construct
quantum codes via the \emph{CSS} or \emph{generalized CSS constructions}
\cite{BDH06}. Moreover, we can easily construct quantum LDPC
codes from classical LDPC codes with girth at least 6. We make use
of classical quasi-cyclic LDPC codes in our construction, and show
that the resulting entanglement-assisted quantum LDPC codes have
good performance; we compare them to examples of standard LDPC
codes already proposed in the literature with similar \emph{net rates},
and show that the new quasicyclic codes have lower block-error rates.

This paper is organized as follows. We discuss properties of binary
circulant matrices, and give a brief introduction to classical
QC-LDPC codes in section \ref{II}. We also prove a few interesting
lemmas regarding classical QC-LDPC codes in this section. In section
\ref{III}, we discuss the principle of constructing quantum QC-LDPC
codes from classical QC-LDPC codes, such that the resulting quantum
QC-LDPC codes require only a small amount of initial pre-shared
entanglement. We also provide two examples of such constructions. In
section \ref{IV}, we compare the performance of the quantum QC-LDPC
codes illustrated in section \ref{III} with some previously proposed
quantum LDPC codes. Finally, in section \ref{V} we conclude.

\section{Preliminary}
\label{II}
\subsection{Properties of binary circulant matrices}
We begin with a well-known proposition for binary circulant
matrices.
\begin{proposition}
\label{iso} The set of binary circulant matrices of size $r \times
r$ forms a ring isomorphic to the ring of polynomials of degree less
than $r$: $\bbF_2[X]/\langle X^r-1\rangle$.
\end{proposition}

Let $M$ be an $r \times r$ circulant matrix over $\bbF_2$. We can
uniquely associate with $M$ a polynomial $M(X)$ with coefficients
given by entries of the first row of $M$. If
$\bc=(c_0,c_1,\cdots,c_{r-1})$ is the first row of the circulant
matrix $M$, then
\begin{equation}
\label{polyM} M(X)=c_0+c_1 X+c_2 X^2+\cdots+c_{r-1}X^{r-1}.
\end{equation}
Adding or multiplying two circulant matrices is equivalent to adding
or multiplying their associated polynomials modulo $X^r-1$. We now
give some useful properties of these matrices and polynomials.

The first lemma is a well-known result in the theory of cyclic codes \cite{FJM77}.
\begin{lemma}
\label{rank} Let $M(X)$ be the polynomial associated with the
$r\times r$ binary circulant matrix $M$. If $\gcd(M(X),X^r-1)=K(X)$,
and the degree of K(X) is $k$, then the rank of $M$ is $r-k$.
\end{lemma}

In the following, we will discuss some particular cases of the circulant matrix $M$ that will play important roles in the later section.
\begin{theorem}
\label{method1}%
Let $r=pq$, where $p,q>1$. Let $\bc=(c_0,c_1,\cdots,c_{r-1})$ be the first row of
an $r\times r$ circulant matrix $M$.
\begin{enumerate}
\item If $c_{(k+pi)\,{\rm mod}\,r}=1$, for some $k$ and $i=0,1,\cdots,(q-1)$, and $0$ otherwise,
then rank$(M)=p$.
\item If $c_{(k+i)\,{\rm mod}\,r}=1$, for some $k$ and $i=0,1\cdots,(p-1)$, and $0$ otherwise,
then rank$(M)=r-p+1$.
\end{enumerate}
\end{theorem}
\begin{proof}
\begin{enumerate}
\item If $c_{(k+pi)\,{\rm mod}\,r}=1$, for some $k$ and
$i=0,1,\cdots,(q-1)$, and $0$ otherwise, we have
\begin{align*}
M(X)&=X^k\left(1+X^p+X^{2p}+\cdots+X^{(q-1)p}\right)\\
&=X^k\left(\frac{X^r-1}{X^p-1}\right).
\end{align*}
Since $\gcd(X^k,X^p-1)=1$, the following holds
\begin{align*}
K(X)&=\gcd\left(M(X),(X^r-1)\right) \\
&=1+X^p+X^{2p}+\cdots+X^{(q-1)p}.
\end{align*}
Then the degree of $K(X)$ is
$pq-p=r-p$. Therefore, by lemma~\ref{rank}, the rank of $M$ is $p$.
\item In this case, the polynomial is

\begin{align*}
M(X)&=X^k\left(1+X+X^2+\cdots+X^{(p-1)}\right) \\
&=X^k\left(\frac{X^r-1}{(X-1)(1+X^p+X^{2p}+\cdots+X^{(q-1)p})}\right)
\end{align*}
and $$K(X)=1+X+X^2+\cdots+X^{(p-1)}.$$
    Then the degree of $K(X)$ is $p-1$.
Therefore, by lemma~\ref{rank}, the rank of $M$ is $r-p+1$.
\end{enumerate}
\end{proof}

We also have the following corollary.
\begin{corollary}
\label{kappa} Suppose $r=pq$, and $p,q>1$. If the weight of each row is $p$, and $K(X)\neq1$,
then the rank $\kappa$ of matrix $M$ is upper-bounded by $r-p+1$.
\end{corollary}
\begin{proof}
Since the weight of $\bc$ is $p$, the lowest possible degree of
$M(X)$ that divides $X^r-1$ is $p-1$, wherein
\begin{equation}\label{eq_ir}%
M(X)=1+X+X^2+\cdots+X^{(p-1)}.
\end{equation}%
Then item 2 of theorem~\ref{method1} confirms
the rank $\kappa$ is at most $r-p+1$.
\end{proof}

\subsection{Classical quasi-cyclic LDPC codes}

\begin{definition}
A binary linear code $C(H)$ of length $n=r\cdot L$ is called
quasi-cyclic (QC) with period $L$ if any codeword which is
cyclically right-shifted by $L$ positions is again a codeword.
Such a code can be represented by a parity-check matrix $H$
consisting of $r\times r$ blocks (by properly rearranging the
coordinates of the code), each of which is an (in general different)
$r\times r$ circulant matrix.
\end{definition}

By the isomorphism mentioned in Prop.~\ref{iso}, we can associate
with each quasi-cyclic parity-check matrix $H\in\bbF_2^{Jr\times
Lr}$ a $J\times L$ polynomial parity-check matrix
$\bH(X)=[h_{j,l}(X)]_{j\in[J],l\in[L]}$ where $h_{j,l}(X)$ is the
polynomial, as defined in (\ref{polyM}), representing the $r\times
r$ circulant submatrix of $H$, and the notation
$[J]:=\{1,2,\cdots,J\}$.

Generally, there are two ways of constructing  $(J,L)$-regular
QC-LDPC by using circulant matrices \cite{SV04}:
\begin{definition}
We say that a QC-LDPC code is Type-I if it is given by a polynomial
parity-check matrix $\bH(X)$ such that all entries are non-zero
monomials. We say that a QC-LDPC code is Type-II if it is given by a
polynomial parity-check matrix $\bH(X)$ with either binomials,
monomials, or zero.
\end{definition}
\subsubsection{Type-I QC-LDPC}

To give an example, let $r=16$, $J=3$, and $L=8$.  The following
polynomial parity check matrix
\begin{equation}
\label{typeI}
\bH(X)=\left[\begin{array}{cccccccc}%
X & X & X & X & X & X & X & X \\
X^2 & X^5 & X^{3} & X^{5} & X^2 & X^5 & X^3 & X^5 \\
X^2 & X^3 & X^4 & X^5 & X^6 & X^7 & X^8 & X^9 \end{array}\right]
\end{equation}
gives a Type-I $(3,8)$-regular QC-LDPC code of length
$n=16\cdot8=128$.
Later on, we will also express $\bH(X)$ by its {\it exponent matrix}
$H_E$. For example, the exponent matrix of (\ref{typeI}) is
\begin{equation}
H_E=\left[\begin{array}{cccccccc}%
1 & 1 & 1 & 1 & 1 & 1 & 1 & 1 \\
2 & 5 & 3 & 5 & 2 & 5 & 3 & 5 \\
2 & 3 & 4 & 5 & 6 & 7 & 8 & 9\end{array}\right].
\end{equation}
The difference of two arbitrary rows of the exponent matrix $H_E$ is
defined as
\begin{equation}
\label{diff}%
\bd_{ij}=\bc_i-\bc_j=\left((c_{i,k}-c_{j,k}) \text{mod}\
r\right)_{k\in[L]},
\end{equation}
where $\bc_i$ is the $i$-th row of $H_E$ and $r$ is the size of the
circulant matrix. We then have
\begin{eqnarray*}
\bd_{21} &=& (1,4,2,4,1,4,2,4)\\
\bd_{31} &=& (1,2,3,4,5,6,7,8) \\
\bd_{32} &=& (0,14,1,0,4,2,5,4).
\end{eqnarray*}
We call an integer sequence $\bd=(d_0,d_1,\cdots,d_{L-1})$ {\it
multiplicity even} if each entry appears an even number of times.
For example, $\bd_{21}$ is multiplicity even, but $\bd_{32}$ is not,
since only $0$ and $4$ appear an even number of times. We call $\bd$
{\it multiplicity free} if no entry is repeated; for example,
$\bd_{31}$.

A simple necessary condition for Type-I $(J,L)$-regular QC-LDPC
codes to give girth $g\geq 6$ is given in \cite{Fossorier04}.
However, a stronger result (both sufficient and necessary condition)
is shown in \cite{HI07QLDPC}. We state these theorems  from
\cite{HI07QLDPC} without proofs.
\begin{theorem}
\label{type-I1} A Type-I QC-LDPC code $C(H_E)$ is dual-containing if
and only if $\bc_i-\bc_j$ is multiplicity even for all $i$ and $j$,
where $\bc_i$ is the $i$-th row of the exponent matrix $H_E$.
\end{theorem}
\begin{theorem}
\label{type-I2} There is no dual-containing Type-I QC-LDPC having
girth $g\geq 6$.
\end{theorem}
\begin{theorem}
\label{type-I3} A necessary and sufficient condition for a Type-I
QC-LDPC code $C(H_E)$ to have girth $g\geq 6$ is $\bc_{i}-\bc_j$ to
be multiplicity free for all $i$ and $j$.
\end{theorem}

\subsubsection{Type-II QC-LDPC}

Take $r=16$, $J=3$, and $L=4$. The following is an example of a
Type-II (3,4)-regular QC-LDPC code:
\begin{equation}
\label{typeII}
\bH(X) = \left[\begin{array}{cccc} X+X^4 & 0 & X^7+X^{10} & 0 \\ X^5 & X^6 & X^{11} & X^{12} \\
0 & X^2+X^{9} & 0 & X^7+X^{13} \end{array}\right].
\end{equation}
The exponent matrix of (\ref{typeII}) is
\begin{equation}
H_E = \left[\begin{array}{cccc} (1,4) & \infty & (7,10) & \infty \\
5 & 6 & 11 & 12 \\
\infty & (2,9) & \infty & (7,13) \end{array}\right].
\end{equation}
Here we denote $X^{\infty}=0$.

The difference of two arbitrary rows of $H_E$ is defined similarly to
(\ref{diff}) with the following additional rules: (i) if for some
entry $c_{i,k}$ is $\infty$, then the difference of $c_{i,k}$ and any
other arbitrary term is again $\infty$; (ii) if the entries
$c_{i,k}$ and $c_{j,k}$ are both binomial, then the difference of
$c_{i,k}$ and $c_{j,k}$ contains four terms. In this example, we
have
\begin{eqnarray*}
\bd_{21} &=& \left((4,1),\infty,(4,1),\infty \right) \\
\bd_{31} &=& \left(\infty,\infty,\infty,\infty \right) \\
\bd_{32} &=& \left(\infty,(12,3),\infty,(11,1)\right) \\
\bd_{11} &=& \left((0,3,13,0),\infty,(0,3,13,0),\infty\right) \\
\bd_{22} &=& \left(0,0,0,0 \right) \\
\bd_{33} &=& \left(\infty,(0,9,7,0),\infty,(0,10,6,0)\right).
\end{eqnarray*}
The definition of {\bf multiplicity even} and {\bf multiplicity free}
is the same, except that we do not take $\infty$ into account. For
example, $\bd_{32}$ is multiplicity free, since there is no pair
with the same entry except $\infty$.
Unlike Type-I QC-LDPC codes, whose $\bd_{ii}$ is always the zero
vector, $\bd_{ii}$ of Type-II QC-LDPC codes can have non-zero
entries.  Therefore it is possible to have cycles of length 4 in a single
layer if $\bd_{ii}$ is not multiplicity free. Each layer is
a set of rows of size $r$ in the original parity check matrix $H$
that corresponds to one row of $H_E$. For example, $\bd_{11}$ is
multiplicity even, therefore the first layer of this Type-II regular
QC-LDPC parity check matrix contains 4-cycles.

In the following, we will generalize theorems
\ref{type-I1}-\ref{type-I3} from the previous section to include
the Type-II QC-LDPC case.
\begin{theorem}
\label{type-II-dual} $C(H_E)$ is a dual-containing Type-II regular
QC-LDPC code if and only if $\bc_i-\bc_j$ is multiplicity even for
all $i$ and $j$.
\end{theorem}
\begin{proof}
Let $\bH(X)=[h_{j,l}(X)]_{j\in[J],l\in[L]}$ be the polynomial parity
check matrix associated with a Type-II $(J,L)$-regular QC-LDPC
parity check matrix $H$. Denote the transpose of $\bH(X)$ by
$\bH(X)^T=[h^t_{l,j}(X)]_{l\in[L],j\in[J]}$, and we have
\begin{equation}
\begin{split}
h^{t}_{l,j}(X)=\begin{cases} 0  & \text{if} \ h_{j,l}(X)=0 \\
X^{r-k} & \text{if} \ h_{j,l}(X)=X^k  \\
X^{r-k_1}+X^{r-k_2} & \text{if} \ h_{j,l}(X)=X^{k_1}+X^{k_2}
\end{cases}.
\end{split}
\end{equation}
Let $\hat{\bH}(X)=\bH(X)\bH(X)^T$, and let the $(i,j)$-th component
of $\hat{\bH}(X)$ be $\hat{h}_{i,j}(X)$. Then
\begin{equation}
\label{Hhat}%
\hat{h}_{i,j}(X)=\sum_{l\in[L]} h_{i,l}(X)h^t_{l,j}(X).
\end{equation}
The condition that $\bd_{ij}$ is multiplicity even implies
that $\hat{h}_{i,j}(X)=0$ modulo $X^r-1$, and vice versa.
\end{proof}

\begin{theorem}
\label{type-II-free} A necessary and sufficient condition for a
Type-II regular QC-LDPC code $C(H_E)$ to have girth $g\geq 6$ is
that $\bc_{i}-\bc_j$ be multiplicity free for all $i$ and $j$.
\end{theorem}
\begin{proof}
The condition that $\bc_{i}-\bc_j$ is multiplicity free for all $i$
and $j$ guarantees that there is no 4-cycle between layer $i$ and layer
$j$, and vice versa.
\end{proof}

\begin{theorem}
There is no dual-containing QC-LDPC having girth $g\geq 6$.
\end{theorem}
\begin{proof}
This proof follows directly from theorem \ref{type-II-dual} and
theorem \ref{type-II-free}. If the Type-II regular QC-LDPC code is
dual-containing, then by theorem \ref{type-II-dual}, $\bc_i-\bc_j$
must be multiplicity even for all $i$ and $j$. However, theorem
\ref{type-II-free} says that this QC-LDPC must contain cycles of
length 4.
\end{proof}

\section{Construction of quantum QC-LDPC codes from classical QC-LDPC codes}
\label{III}

It has been shown that any classical linear binary or quaternary code
can be used to construct a corresponding entanglement-assisted
quantum error-correcting code \cite{BDH06,BDH06IEEE}.

\begin{theorem}
\label{ebit} Let $C(H)$ be a binary classical $[n,k,d]$ code with
parity check matrix $H$. We can obtain a corresponding
$[[n,2k-n+c,d;c]]$ EAQECC, where $c = {\rm rank}(H H^T)$ is the
number of ebits needed.
\end{theorem}
\begin{proof}
See \cite{HBD07}.
\end{proof}

\noindent {\bf Remark}\,\, An $[[n,k',d;c]]$ EAQECC encodes $k'$
logic qubits into $n$ physical qubits with the help of $c$ ebits
($c$ copies of maximally entangled states). We define the \emph{net
rate} of such an EAQECC to be $\frac{k'-c}{n}$. The definition of net
rate only takes account of the effective qubits that are sent
through a channel if we trade the stronger resource of pure identity
qubit channel with the weaker resource of pure entanglement. When
$c=0$, this quantity is equal to the ``rate'' of a standard QECC.
Another way to think of this is that the net rate is the rate we can
achieve if we ``borrow'' $c$ ebits in order to send the codeword,
then ``pay them back'' by using $c$ communication qubits to
establish a new $c$ ebits.  Because this catalytic mode makes no
net consumption of ebits, it is quite reasonable to compare the
net rate of an EAQECC to the rate of a standard QECC.

In the following, we will consider conditions that will give us
$(J,L)$-regular QC-LDPC codes $C(H)$ with girth $g\geq 6$ and with
the rank of $H H^T$ as small as possible. Let
$\hat{\bH}(X)=\bH(X)\bH(X)^T$ be the polynomial representation of $H
H^T$. In general, $\hat{\bH}(X)$ represents a square symmetric
matrix $\hat{H}$ with size $Jr\times Jr$ that contains $J^2$
circulant $r\times r$ matrices represented by $\hat{h}_{i,j}(X)$ as
defined in (\ref{Hhat}). Next, we provide two examples to illustrate
two different ways of minimizing the rank of the square symmetric
matrix represented by $\hat{\bH}(X)$. This would minimize the number
of ebits when we use the classical code $C(H)$ to construct the EAQECC.


The first method is to make the matrix $\hat{H}=HH^T$ become a
circulant matrix with a small rank. This can be achieved by properly choosing
$\bH(X)$ such that the elements $\hat{h}_{i,j}(X)$ in $\hat{\bH}(X)$ satisfying:
\begin{equation}%
\label{cond_cir}%
\hat{h}_{i,j}(X)=\hat{h}_{i+1,j+1}(X),
\end{equation}%
for $i,j=0,1,\cdots,J-2.$ First notice that each polynomial matrix $\hat{h}_{i,j}(X)$ in $\hat{\bH}(X)$ is a circulant matrix of size $r\times r$, and the polynomial matrix $\hat{\bH}(X)$ contains $J^2$ such circulant matrices. Since condition (\ref{cond_cir}) guarantees $\hat{H}$ is itself  circulant, $\hat{\bH}(X)$ can be represented by some polynomial $g(X)$ in $\bbF_2[X]/\langle X^{Jr}-1\rangle$. The rank $\kappa$ of $\hat{H}$ can then be read off by lemma~\ref{rank}: If $\gcd(g(X),X^{Jr}-1)=K(X)$,
and the degree of $K(X)=k$, then $\kappa=Jr-k$.

Let's look at an example of this type using a classical Type-I
QC-LDPC code. Consider $r=16$, $J=3$, $L=8$, and the following
polynomial parity check matrix $\bH(X)$:
\begin{equation}
\label{ex1}
\bH(X)=\left[\begin{array}{cccccccc} X & X & X & X & X & X & X & X \\
X & X^2 & X^3 & X^4 & X^5 & X^6 & X^7 & X^8 \\
X & X^3 & X^5 & X^7 & X^9 & X^{11} & X^{13} & X^{15}
\end{array}\right].
\end{equation} 
Simple calculation shows that
\begin{equation}%
\label{ex1_hx}
\hat{h}_{i,j}(X)=\begin{cases}0, &\text{$i=j$},  \\
\sum_{k=0}^{7}X^k,  & i=j+1 \\ \sum_{k=0}^{7}X^{2k}. & i=j+2
\end{cases}
\end{equation}
Since (\ref{ex1_hx}) satisfies (\ref{cond_cir}), $\hat{\bH}(X)$
represents a circulant matrix, and the polynomial associated with
$\hat{H}$ is
\[
g(X) = X^{16}\left(\sum_{k=0}^{7}X^k\right) +
X^{32}\left(\sum_{k=0}^{7}X^{2k}\right).
\]
The degree of $\gcd(g(X),X^{48}-1)=30$, therefore by
lemma~\ref{rank}, the number of ebits that was needed to construct
the corresponding quantum code is only 18. Actually, (\ref{ex1})
gives us a $[[128,58,6;18]]$ EAQECC, and we will refer to this example
as ``Ex1'' later in section \ref{IV}.  The {\it net rate} of this code is
$(k-c)/n=40/128$.

{\bf Remark}\,\, The parity check matrix $H$ of
(\ref{ex1}) gives a $[128,84,6]$ classical code. The rate of the
QC-LDPC code is actually slightly higher than $(L-J)/L$,
since $H$ usually contains linearly dependent rows. For example,
each layer of $H$ contains the all-one vector; therefore, we can
find at least ($J-1$) linearly dependent rows in $H$.

The second method is to minimize the rank of each circulant matrix
inside $\hat{H}$, that is, to minimize the rank of the circulant
matrix represented by $\hat{h}_{i,j}(X)$, $\forall i,j$.
Let the rank of $\hat{H}$ be $\kappa$. Then
\begin{equation}
\label{u_rank}%
\kappa \leq \sum_{i=1}^{J} \max_{j\in[J]}{\kappa_{i,j}}.
\end{equation}
This upper bound is not tight in general, e.g., when $L$ is odd in
the Type-I $(J,L)$-regular QC-LDPC codes, the bound of
(\ref{u_rank}) gives $Jr$, which is equal to the number of rows of
$\hat{H}$. This is because $\kappa_{i,i}=r$ for every $i$. However,
with certain restrictions (e.g. even $L$), we can obtain a
reasonable upper bound for $\kappa$.

\begin{theorem}
\label{bound} Given a $(J,L)$-regular QC-LDPC code with polynomial
parity check matrix $\bH(X)$ such that the exponent matrix of
$\hat{\bH}(X)=\bH(X)\bH(X)^T$ is multiplicity free, if whenever
$\hat{h}_{i,j}(X)\neq 0$, $\gcd(\hat{h}_{i,j}(X),X^r-1)\neq1$, then
the rank $\kappa$ is upper bounded by $J(r-L+1)$.
\end{theorem}%
\begin{proof}%
Let $\hat{h}_{i,j}$ be the circulant matrix associated with the
polynomial $\hat{h}_{i,j}(X)$. Since the exponent matrix of
$\hat{\bH}(X)$ is multiplicity free, then the weight of each row
vector of $\hat{h}_{i,j}$ is $L$ whenever $\hat{h}_{i,j}(X)\neq0$.
By corollary~\ref{kappa}, $\kappa_{i,j} \leq r-L+1$. Therefore
\begin{equation}%
\label{tight_bound}%
\kappa\leq \sum_{i=1}^{J} \max_{j \in[J]} \kappa_{i,j} \leq J(r-L+1).
\end{equation}%
\end{proof}

Define the entanglement consumption rate to be $\kappa/n$.
Since the rank decides the number of EPR pairs that are required by the corresponding entanglement-assisted quantum code, we want $\kappa/n$ to be as small as possible.
It is easy to see that this bound becomes tighter when we pick $L$ much larger than $J$:
\begin{equation*}
\frac{\kappa}{n}\leq\frac{J(r-L+1)}{n}\leq \frac{Jr}{Lr} =\frac{J}{L}.
\end{equation*}

In the following, we present an example showing that the restriction
in theorem \ref{bound} is achievable. This example comes from a
classical Type-II QC-LDPC code. Consider $r=16$, $J=3$, $L=8$, and
the following polynomial parity check matrix $\bH(X)$:
\begin{widetext}
\begin{equation}
\label{ex2}%
\bH(X) = \left[\begin{array}{cccccccc} X+X^2 & 0 & X+X^4 & 0 & X+X^6 & 0 & X+X^8 & 0 \\
X^5 & X^5 & X^6 & X^6 & X^7 & X^7 & X^8 & X^8 \\
0 & X+X^2 & 0 & X+X^4 & 0 & X+X^6 & 0 & X+X^8 \end{array}\right].
\end{equation}
\end{widetext}%
Simple calculation shows that
\begin{equation}%
\label{ex2_hx}%
\hat{h}_{i,j}(X) = \begin{cases}0, &(i,j)=(2,2),(1,3),\text{or} (3,1) \\
\sum_{k=0}^{7}X^{1+2k}, &(i,j)=(1,1),(3,3) \\
\sum_{k=0}^{7}X^k. &(i,j)=(2,1),(2,3)
\end{cases}%
\end{equation}%
In this example, (\ref{ex2_hx}) satisfies the statement given in
theorem \ref{bound}. Therefore, the rank of $\hat{\bH}(X)$ is upper
bounded by $27$. The polynomial parity check matrix in
(\ref{ex2}) gives a $[[128,58,6;18]]$ quantum QC-LDPC code, and we
will refer to this example as ``Ex2'' in section \ref{IV}.  It also has net
rate $(k-c)/n=40/128$, just like Ex1.

{\bf Remark}\,\, Though no general guidelines of using
these two methods to construct the desired polynomial parity check
matrix $\bH(X)$ are given, examples can be obtained with a simple
search even when the code length is very long. This is because if we
want to construct classical QC-LDPC codes with long length, we will
increase the parameter $r$ rather than $L$ and $J$. Choosing larger
$r$ increases the distance property of the QC-LDPC codes
\cite{SV04}. Therefore, in general, $L$ and $J$ are not very big,
which makes the search not difficult to perform.

\section{Performance}
\label{IV}%
In this section, we compare the performance of the quantum LDPC codes
given in Sec.~\ref{III} to simulation results for two constructions currently
available in the dual-containing quantum LDPC codes literature. The
criterion for comparison between these quantum LDPC codes is the net
rate (this is equal to the definition of ``rate'' for standard
QECCs). If the quantum codes are constructed by the CSS
construction  from classical QC-LDPC codes with same parameters $J$
and $L$, the corresponding standard QECCs and EAQECCs have almost
same rate, which equals $\frac{L-2J}{L}$. Slight differences
in net rates are possible because of the possibility of different
numbers of linearly dependent rows in two parity check matrices of
classical QC-LDPC codes.

The authors in \cite{MMM04QLDPC} proposed four dual-containing constructions of quantum LDPC codes. Among those constructions, we use one construction in particular (Ref. \cite{MMM04QLDPC} called it construction ``B'') as a benchmark, since it performed the best of their constructions, especially in the low quantum rate (less than 0.5) and medium code length (less than 10000 qubits) case. Comparison of the performance of all these 4 constructions is illustrated in Fig.~17 on page~35 of \cite{MMM04QLDPC}. The construction is as follows: take an $n/2 \times n/2$ cyclic matrix $C$ with row weight $L/2$, and define
\[
H_0=[C,C^T].
\]
We then delete some rows from $H_0$ to obtain a matrix $H$ with $m$
rows. It is easy to verify that $H$ is dual-containing. Therefore, by
the CSS construction \cite{Ste96,shor95}, we can obtain standard
quantum LDPC codes of length $n$. The advantage of this construction
is that the choice of $n,m$, and $L$ is completely flexible;
however, the column weight $J$ is not fixed. We picked $n=128$,
$m=48$, and $L=8$, and called this [[128,32]] quantum LDPC code
``Ex-MacKay.''  It has rate $k/n=32/128$, lower than the net rate of
Ex1 and Ex2.

The second example is a standard quantum LDPC code that was
constructed from classical QC-LDPC codes \cite{HI07QLDPC}. This
construction is the first example of standard quantum LDPC codes
with no 4-cycles.
\begin{theorem}
\label{HIcont}%
Let $P$ be an integer which is greater than 2 and $\sigma$ an
element of $\bbZ_P^*:=\{z:z^{-1} \text{exists}\}$ with
$ord(\sigma)\neq |\bbZ_P^*|$, where
$ord(\sigma):=\min\{m>0|\sigma^m=1\}$ and $|X|$ means the
cardinality of a set $X$. If we pick any $\tau\in\bbZ_P^*=
\{1,\sigma,\sigma^2,\cdots\}$, define
\begin{eqnarray*}
c_{j,l}&:=& \begin{cases} \sigma^{-j+l} & 0\leq l < L/2
\\-\tau\sigma^{j-1+l} & L/2\leq l < L \end{cases} \\
d_{k,l}&:=& \begin{cases} \tau\sigma^{-k-1+l} & 0\leq l<L/2 \\
-\sigma^{k+l} &L/2\leq l < L \end{cases},
\end{eqnarray*}
and define the exponent matrix $H_C$ and $H_D$ as
\[
H_C=[c_{j,l}]_{j\in[J],l\in[L]},\ \
H_D=[d_{k,l}]_{k\in[K],l\in[L]},
\]
where $L/2=ord(\sigma)$ and $1\leq J,K \leq L/2$, then $H_C$ and
$H_D$ can be used to construct quantum QC-LDPC codes with girth at
least 6.
\end{theorem}

Here, we pick the set of parameters $(J,L,P,\sigma,\tau)$ to be
$(3,8,15,2,3)$, to get a code with similar block size and rate to the other
examples. The exponent matrices $H_C$ and $H_D$ described in
theorem \ref{HIcont} are
\begin{eqnarray}
\label{HcHd}
H_C &=& \left[\begin{array}{cccccccc}%
1 & 2 & 4 & 8 & 6 & 12 & 9 & 3 \\
8 & 1 & 2 & 4 & 12 & 9 & 3 & 6 \\
4 & 8 & 1 & 2 & 9 & 3& 6 & 12
\end{array}\right] \\
H_D &=& \left[\begin{array}{cccccccc}%
9 & 3 & 6 & 12 & 14 & 13 & 11 & 7 \\
12& 9 & 3 & 6 & 13 & 11 & 7 & 14 \\
6 & 12 & 9 & 3 & 11 & 7 & 14 & 13
\end{array}\right],
\end{eqnarray}
and by the CSS construction, it will give a $[[120,38,4]]$ quantum
QC-LDPC code. We will call this code ``Ex-HI''.  It has rate
$k/n=38/120$, slightly higher (just over 1\%) than Ex1 and Ex2.

\begin{figure}[htbp]
    \includegraphics[width=0.5\textwidth]{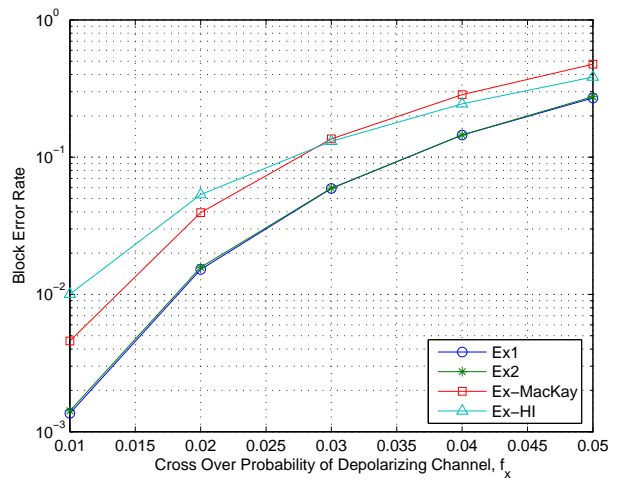}
  \caption{(Color online). Performance of quantum LDPC codes with SPA decoding, and 100-iteration}
  \label{fig}
\end{figure}

In the simulation, the channel is assumed to be the depolarizing
channel, which creates $X$ errors, $Y$ errors, and $Z$ errors with
equal probability $f_m$. Since all these quantum LDPC codes are
CSS-type quantum codes, $Z$ errors and $X$ errors can be decoded and
corrected separately by the sum-product algorithm using a standard
classical correction algorithm  \cite{MMM04QLDPC}.
Because the depolarizing channel can be thought of as producing a
separate list of $X$ errors and $Z$ errors (with a $Y$ error being both
an $X$ and a $Z$), doing two separate ``classical'' correction steps gives
an accurate simulation of quantum error correction.  For an EAQECC we
must include Bob's half of the shared ebits in the error correction step;
these bits, not having passed through the channel, are of course error-free.

We compare the performance of our examples in section \ref{III} with
these two dual-containing quantum LDPC codes in figure \ref{fig}.
The performances of Ex1 and Ex2 do not differ much. This is not
surprising, since these two codes have similar parameters.
The performance of Ex-MacKay is worse than Ex1 and Ex2, probably
because there are so many 4-cycles in Ex-MacKay. These cycles impair
the performance of sum-product decoding algorithm.  The
entanglement-assisted quantum QC-LDPC codes also outperform the
quantum QC-LDPC code of Ex-HI, probably because the classical QC-LDPC
codes used to construct our examples have better distance properties than
the classical QC-LDPC of ex-HI. This simulation result is also consistent with
our result in \cite{BDH06}: better classical codes give better quantum codes.

Of course, these 4 simulation results are by no means an exhaustive study of
all possible quantum LDPC codes.  These four examples were chosen without
optimizing for their error-correcting performance; in general, one expects generically
similar results for two codes produced by the same construction, so these comparisons
are likely to be typical.  This does not, of course, argue that there are no standard
quantum LDPC codes, or other entanglement-assisted LDPC  codes for that matter,
with performance superior to those in this paper.  However, in general the results
match our expectations:  it is much easier to find classical codes with good performance
(i.e., no 4-cycles and good distance properties), and from them construct quantum
codes that use a relatively small amount of shared entanglement,
than to satisfy the exact constraint of a dual-containing code; even more so than to
simultaneously find a code that is dual-containing and contains no 4-cycles.

Moreover, since these codes use the CSS construction, it is not surprising that
the iterative decoding algorithm works well:  we are effectively doing two successive
classical decoding steps, using classical codes that are known to be good, and
reflecting the classical result that codes without 4-cycles tend to outperform codes
with 4-cycles.  Furthermore, Devetak's proof of the quantum channel coding theorem
shows that codes with a CSS-like structure are good enough to achieve capacity
\cite{Devetak03}.  Therefore, it is quite possible that CSS-type quantum LDPC codes
are sufficient to given performance as good as is practically possible.  (Though we
certainly do not rule out the possibility that studying general additive LDPC codes over
GF(4) would reveal interesting properties---that is work for the future.)

\section{Conclusions}
\label{V}%
There are two advantages of Type-II QC-LDPCs over Type-I QC-LDPCs.
First, according to \cite{SV04} certain configurations of Type-II QC-LDPC
codes have larger minimum distance than Type-I QC-LDPC. Therefore,
we can construct better quantum QC-LDPCs from classical Type-II
QC-LDPC codes. Second, it seems likely that Type-II QC-LDPCs will have more
flexibility in constructing quantum QC-LDPC codes with small amount
of pre-shared entanglement, because of the ability to insert zero
submatrices. However, further investigation of this issue is required.

By using the entanglement-assisted error correction formalism, it is
possible to construct EAQECCs from any classical linear code, not just
dual-containing codes.  We have shown how to do this for two classes of
quasi-cyclic LDPC codes (Type-I and Type-II), and proven a number of
theorems that make it possible to bound how much entanglement is required
to send a code block for codes of these types. Using these results, we have
been able to easily construct examples of quantum QC-LDPC codes that
require only a relatively small amount of initial shared entanglement, and that
perform better (based on numerical simulations) than examples of previously
constructed dual-containing quantum LDPC codes. Since in general the
properties of quantum codes follows directly from the properties of the
classical codes used to construct them, and the evidence of our
examples suggests that the iterative decoders can also be made to
work effectively on the quantum versions of these codes, this should
make possible the construction of large-scale efficient quantum
codes.  These codes could be useful for quantum communications;
conceivably they could also be used as building blocks for standard
QECCs that might be of use in quantum computation, though that is
still a subject for further research.

We are especially interested in developing a new quantum decoding
algorithm in the future. Though the SPA decoding algorithm gives
a reasonable trade-off between complexity and performance, it
may not be the best choice for decoding quantum errors. One
reason for concern is that SPA ignores the purely quantum phenomenon
of degeneracy in the decoding process, which could possibly result in
introducing more errors instead of correcting them. Though this issue can
hopefully be fixed by adding simple heuristic methods on SPA, degeneracy
has also been shown to lead to convergence problems for some codes
\cite{PC08QLDPC}, though we have not observed that effect on performance
for the codes we present in this paper. If convergence problems prove to be common, we hope these can be overcome by a true quantum decoding algorithm.


\section*{Acknowledgments}
M.H. and T.A.B. acknowledge financial supported from NSF Grant
No.~CCF-0448658. I.D. and M.H. acknowledge financial support from
NSF Grant No.~CCF-0524811 and NSF Grant No.~CCF-0545845.

\bibliography{Ref}

\begin{thebibliography}{10}

\bibitem{RG63thesis}
R.~G. Gallager.
\newblock {\em Low-Density Parity-Check Codes}.
\newblock PhD thesis, Massachusetts Institute of Technology, 1963.

\bibitem{MN96}
D.~J.~C. MacKay and R.~M. Neal.
\newblock Near shannon limit performance of low density parity check codes.
\newblock {\em Electronic Letters}, 32(18):1645--1646, 1996.

\bibitem{davey98low}
M.~C. Davey and D.~J.~C. MacKay.
\newblock Low density parity check codes over {GF}(q).
\newblock {\em IEEE Communications Letters}, 2:165--167, 1998.

\bibitem{mackay99good}
D.~J.~C. MacKay.
\newblock Good error-correcting codes based on very sparse matrices.
\newblock {\em IEEE Trans. Inf. Theory}, 45:399--432, 1999.

\bibitem{Mackay98turbo}
D.~J.~C. MacKay.
\newblock Gallager codes that are better than turbo codes.
\newblock {\em Proc. 36th Allerton Conf. Communication, Control, and
  Computing}, 1998.
\newblock Monticello, IL.

\bibitem{KLF01}
Y.~Kou, S.~Lin, and M.~Fossorier.
\newblock Low-density parity-check codes based on finite geometries: A
  rediscovery and new results.
\newblock {\em IEEE Trans. Inf. Theory}, 47:2711--2736, 2001.

\bibitem{Fossorier04}
M.~Fossorier.
\newblock Quasi-cyclic low-density parity-check codes from circulant
  permutation matrices.
\newblock {\em IEEE Trans. Inf. Theory}, 50(8):1788--1793, 2004.

\bibitem{Tanner81}
R.~M. Tanner.
\newblock A recursive approach to low complexity codes.
\newblock {\em IEEE Trans. Inf. Theory}, 27:533--547, 1981.

\bibitem{HI07QLDPC}
Manabu Hagiwara and Hideki Imai.
\newblock Quantum quasi-cyclic ldpc codes.
\newblock In {\em Proceedings of the IEEE International Symposium on
  Information Theory}, pages 806--810, June 2007.

\bibitem{MMM04QLDPC}
D.~J.~C. MacKay, G.~Mitchison, and P.~L. McFadden.
\newblock Sparse-graph codes for quantum error correction.
\newblock {\em IEEE Trans. Inf. Theory}, 50:2315--2330, 2004.

\bibitem{COT05QLDPC}
T.~Camara, H.~Ollivier, and J.-P. Tillich.
\newblock Constructions and performance of classes of quantum ldpc codes, 2005.
\newblock quant-ph/0502086.

\bibitem{PC08QLDPC}
David Poulin and Yeojin Chung.
\newblock On the iterative decoding of sparse quantum codes.
\newblock {\em Quantum Information and Computation}, 8(10):987--1000, 2008.

\bibitem{BDH06}
T.~Brun, I.~Devetak, and M.~H. Hsieh.
\newblock Correcting quantum errors with entanglement.
\newblock {\em Science}, 314(5798):436--439, 2006.

\bibitem{BDH06IEEE}
T.~Brun, I.~Devetak, and M.~H. Hsieh.
\newblock Catalytic quantum error correction, 2006.
\newblock quant-ph/0608027.

\bibitem{FJM77}
F.J. MacWilliams and N.J.A. Sloane.
\newblock {\em The Theory of Error-Correcting Codes}.
\newblock Elsevier, Amsterdam, 1977.

\bibitem{SV04}
R.~Smarandache and P.~O. Vontobel.
\newblock On regular quasi-cyclic ldpc codes from binomials.
\newblock In {\em Proceedings of the IEEE International Symposium on
  Information Theory}, page 274, July 2004.

\bibitem{HBD07}
M.~H. Hsieh, I.~Devetak, and T.~Brun.
\newblock General entanglement-assisted quantum error-correcting codes.
\newblock {\em Phys. Rev. A}, 76:062313, 2007.

\bibitem{Ste96}
A.~M. Steane.
\newblock Error-correcting codes in quantum theory.
\newblock {\em Phys. Rev. Lett.}, 77:793--797, 1996.

\bibitem{shor95}
P.~W. Shor.
\newblock Scheme for reducing decoherence in quantum computer memory.
\newblock {\em Phys. Rev. A}, 52:R2493, 1995.

\bibitem{Devetak03}
I.~Devetak.
\newblock The private classical capacity and quantum capacity of a quantum
  channel.
\newblock {\em IEEE Trans. Inf. Theory}, 51(1):44--55, 2005.

\end{thebibliography}
\bibliographystyle{unsrt}

\end{document}